\documentclass[letterpaper, 12pt]{article}

\newenvironment{proof}{\paragraph{\bf Proof:}}{\hspace*{\fill}\(\Box\)}

\title{Matrix Representation of Iterative Approximate Byzantine Consensus in Directed Graphs\thanks{This research is supported
in part by
 National
Science Foundation award CNS 1059540 and Army Research Office grant W-911-NF-0710287. Any opinions, findings, and conclusions or recommendations expressed here are those of the authors and do not
necessarily reflect the views of the funding agencies or the U.S. government.}
}

\author{Nitin Vaidya\\ Department of Electrical and Computer Engineering\\ University of Illinois at Urbana-Champaign\\ nhv@illinois.edu}

\usepackage{graphicx}
\usepackage{amsmath}

\newcommand{\comment}[1]{}

\usepackage{graphicx}
\usepackage{latexsym}

\newcommand{\scriptf}{\mathcal{F}}
\newcommand{\scripte}{\mathcal{E}}
\newcommand{\scriptv}{\mathcal{V}}

\newtheorem{theorem}{Theorem}

\newtheorem{claim}{Claim}

\newtheorem{definition}{Definition}

\newtheorem{lemma}{Lemma}

\def\noflash#1{\setbox0=\hbox{#1}\hbox to 1\wd0{\hfill}}

\begin{document}

\date{~}

\maketitle

\centerline{March 8, 2012}

~

\begin{abstract}
This paper presents a proof of correctness of an iterative
approximate Byzantine
consensus (IABC) algorithm for directed graphs. The iterative algorithm allows
fault-free nodes to reach approximate conensus despite the presence of up
to $f$ Byzantine faults. Necessary conditions on the underlying network
graph for the existence of a correct IABC algorithm were shown in our
recent work \cite{IBA_sync,us}. \cite{IBA_sync} also analyzed
a specific IABC algorithm and showed that it performs correctly in any
network graph that satisfies the necessary condition, proving that the
necessary condition is also sufficient. In this paper, we present an 
alternate proof of correctness of the IABC algorithm, using a familiar
technique based on transition matrices \cite{jadbabaie_consensus,Benezit,vaidyaII,Zhang}.\\

The key contribution of this paper is to exploit the
following observation: for a {\em given} evolution of the state vector
corresponding to the state of the fault-free nodes, 
many alternate state transition matrices may be chosen to model
that evolution correctly. For a given state evolution, we identify one approach
to suitably ``design'' the transition matrices so that the standard tools
for proving convergence can be applied to the Byzantine fault-tolerant
algorithm as well.
In particular, the transition matrix for each iteration
is designed such that each row of the matrix contains a large enough number
of elements that are bounded away from 0. \\

\end{abstract}


\newcommand{\deltaC}{\delta_{N_i^*[t]}}
\newcommand{\bfA}{{\bf A}}
\newcommand{\bfB}{{\bf B}}
\newcommand{\bfC}{{\bf C}}
\newcommand{\bfG}{{\bf G}}
\newcommand{\bfH}{{\bf H}}
\newcommand{\bfK}{{\bf K}}
\newcommand{\bfM}{{\bf M}}
\newcommand{\bfP}{{\bf P}}
\newcommand{\bfQ}{{\bf Q}}
\newcommand{\bfv}{{\bf v}}
\newcommand{\sH}{\mathcal{H}}
\newcommand{\T}[1]{\overline{#1}}

\section{Introduction}
\label{sec:intro}

Dolev et al. \cite{AA_Dolev_1986} introduced the notion of
{\em approximate Byzantine consensus} by relaxing the requirement
of {\em exact} consensus \cite{AA_nancy}.
The goal in approximate consensus is to allow the fault-free nodes to agree on values that are approximately equal to each other (and {\em not necessarily}
exactly identical). 
In presence of Byzantine faults, while {\em exact} consensus 
is impossible in {\em asynchronous} systems \cite{FLP_one_crash}, approximate
consensus is achievable \cite{AA_Dolev_1986}.
The notion of approximate consensus is of interest in {\em synchronous}
systems as well, since approximate consensus can be achieved using
simple distributed algorithms that do {\em not} require complete knowledge of
the network topology \cite{AA_convergence_markov}.

In this paper, we are interested in iterative algorithms
for achieving approximate Byzantine consensus in synchronous point-to-point
networks that are modeled by arbitrary {\em directed}\, graphs.
The {\em iterative
approximate Byzantine consensus} (IABC) algorithms of interest have
the following properties, which we will soon state more formally:
\begin{itemize}
\item {\em Initial state} of each node is equal to a real-valued
{\em input} provided to that node.
\item {\em Validity} condition: After each iteration of an IABC algorithm, the state of each fault-free node
must remain in the {\em convex hull} of the states of the fault-free nodes
at the end of the {\em previous} iteration.
\item {\em Convergence} condition:
For any $\epsilon>0$, after a sufficiently large number of iterations,
the states of the fault-free nodes are guaranteed to be within $\epsilon$
of each other.
\end{itemize}
Certain IABC algorithms have been shown to satisfy the above properties
in {\em fully
connected} graphs \cite{AA_Dolev_1986,AA_nancy}, and in {\em arbitrary
directed} graphs satisfying a tight necessary condition 
\cite{IBA_sync,us}.
Please refer to \cite{IBA_sync,us} for a summary of the related work.

The main contribution of this paper is to develop an alternate proof of
correctness for a IABC algorithm, which was proved correct
in arbitrary graphs that satisfy a necessary condition developed
in our prior work \cite{IBA_sync}. 
The alternate proof is based on
transition matrices that capture the behavior of the IABC algorithm
executed by the fault-free nodes.  This work is inspired
by, and borrows some matrix analysis tools from, other work that also uses
transition matrices in 
related contexts \cite{jadbabaie_consensus,Benezit,vaidyaII,Zhang}.
This paper exploits the
following observation: for a {\em given} evolution of the state vector
corresponding to the state of the fault-free nodes, 
many alternate state transition matrices may potentially be chosen to emulate
that evolution correctly. For a given state evolution, we identify one approach
to suitably ``design'' the transition matrices so that the standard tools
can be applied to prove convergence of the Byzantine fault-tolerant
algorithm in {\em all networks} that satisfy a necessary condition
(proved in \cite{us}) on the network communication graph.
In particular, the transition matrix for each iteration
is designed such that each row of the matrix contains a large enough number
of elements that are bounded away from 0.

\section{Network and Failure Models}

\paragraph{Network Model:}
The system is assumed to be {\em synchronous}.
The communication network is modeled as a simple {\em directed} graph $G(\scriptv,\scripte)$, where $\scriptv=\{1,\dots,n\}$ is the set of $n$ nodes, and $\scripte$ is the set of directed edges between the nodes in $\scriptv$. 
 Node $i$ can reliably transmit messages to node $j$ if and only if
the directed edge $(i,j)$ is in $\scripte$.
Each node can send messages to itself as well, however,
for convenience, we {\em exclude self-loops} from set $\scripte$.
That is, $(i,i)\not\in\scripte$ for $i\in\scriptv$.
With a slight abuse of terminology, we will use the terms {\em edge}
and {\em link} interchangeably in our presentation.

For each node $i$, let $N_i^-$ be the set of nodes from which $i$ has incoming
edges.
That is, $N_i^- = \{\, j ~|~ (j,i)\in \scripte\, \}$.
Similarly, define $N_i^+$ as the set of nodes to which node $i$
has outgoing edges. That is, $N_i^+ = \{\, j ~|~ (i,j)\in \scripte\, \}$.
Since we exclude self-loops from $\scripte$,
$i\not\in N_i^-$ and $i\not\in N_i^+$. 
However, we note again that each node can indeed send messages to itself.
A necessary condition for correctness of an IABC algorithm for $f>0$ is that
$|N_i^-|>2f$ \cite{IBA_sync}.

Node $j$ is said to be an {\em incoming neighbor} of node $i$,
if $j\in N_i^-$. Similarly, $j$ is said to be an {\em outgoing neighbor}
of node $i$, if $j\in N_i^+$.


\paragraph{Failure Model:}
We consider the Byzantine failure model, with up to $f$ nodes becoming faulty. A faulty node may {\em misbehave} arbitrarily. Possible misbehavior includes sending incorrect and mismatching (or inconsistent) messages to different neighbors. The faulty nodes may potentially collaborate with each other. Moreover, the faulty nodes are assumed to have a complete knowledge of the execution of
the algorithm, including the states of all the nodes,
contents of messages the other nodes send to each other,
the algorithm specification, and the network topology.

\section{Iterative Approximate Byzantine Consensus (IABC)}
\label{sec:iabc}

Each node $i$ maintains state $v_i$, with $v_i[t]$ denoting the state
of node $i$ at the {\em end}\, of the $t$-th iteration of the algorithm.
Initial state of node $i$,
$v_i[0]$, is equal to the initial {\em input}\, provided to node $i$.
At the {\em start} of the $t$-th iteration ($t>0$), the state of
node $i$ is $v_i[t-1]$.

Let $\scriptf$ denote the set of faulty nodes.
Thus, the nodes
in $\scriptv-\scriptf$ are non-faulty.\footnote{\normalsize For sets $X$ and $Y$, $X-Y$ contains elements that are in $X$ but not in $Y$. That is, $X-Y=\{i~|~ i\in X,~i\not\in Y\}$.} 
\begin{itemize}

\item $U[t] = \max_{i\in\scriptv-\scriptf}\,v_i[t]$. $U[t]$ is the largest state among the fault-free nodes at the end of the $t$-th iteration.
Since the initial state of each node is equal to its input,
$U[0]$ is equal to the maximum value of the initial input at the fault-free nodes.

\item $\mu[t] = \min_{i\in\scriptv-\scriptf}\,v_i[t]$. $\mu[t]$ is the smallest state among the fault-free nodes at the end of the $t$-th iteration.
$\mu[0]$ is equal to the minimum value of the initial input at the
fault-free nodes.
\end{itemize}
The following conditions must be satisfied by an IABC algorithm
in presence of up to $f$ Byzantine faulty nodes:
\begin{itemize}
\item {\em Validity:} $\forall t>0,
~~\mu[t]\ge \mu[t-1]
~\mbox{~~and~~}~
~U[t]\le U[t-1]$

\item {\em Convergence:} $\lim_{\,t\rightarrow\infty} ~ U[t]-\mu[t] = 0$.
Equivalently, $\lim_{\,t\rightarrow\infty} ~ v_i[t]-v_j[t] = 0$, for
$i,j\in \scriptv-\scriptf$.
\end{itemize}

An iterative algorithm is said to be {\em correct} if it satisfies
the {\em validity} and {\em convergence} conditions.
We will prove the correctness of Algorithm 1 below
 in all graphs that satisfy the
necessary condition in Theorem 2 of \cite{us}. The algorithm
should be performed by each node $i$ in the
$t$-th iteration, $t\geq 1$. The faulty nodes may deviate from the
algorithm specification. If a fault-free node does not receive an
expected message from an incoming neighbor (in the {\em Receive step}
below), then that message is assumed to have some default value.

\vspace*{8pt}\hrule

{\bf Algorithm 1}
\vspace*{4pt}\hrule

~

Steps to be performed by node $i$ in the $t$-th iteration:
\begin{enumerate}

\item {\em Transmit step:} Transmit current state $v_i[t-1]$ on all outgoing edges.
\item {\em Receive step:} Receive values on all incoming edges. These values form
vector $r_i[t]$ of size $|N_i^-|$.


\item {\em Update step:}
Sort the values in $r_i[t]$ in an increasing order, and eliminate
the smallest $f$ values, and the largest $f$ values (breaking ties
arbitrarily).
 Let $N_i^*[t]$ denote the identifiers of nodes from
whom the remaining $|N_i^-| - 2f$ values were received, and let
$w_j$ denote the value received from node $j\in N_i^*[t]$.

For convenience, define $w_i=v_i[t-1]$.

Observe that
if $j\in \{i\}\cup N_i^*[t]$ is fault-free, then $w_j=v_j[t-1]$.

Define
\begin{eqnarray}
v_i[t] ~ = ~\sum_{j\in \{i\}\cup N_i^*[t]} a_i \, w_j
\label{e_Z}
\end{eqnarray}
where
\[ a_i ~=~ \frac{1}{|N_i^-|-2f+1} ~=~
 \frac{1}{|N_i^*[t]|+1}
\] 
Recall that
 $i\not\in N_i^*[t]$
because $(i,i)\not\in\scripte$.
The ``weight'' of each term on the right-hand side of
(\ref{e_Z}) is $a_i$, and these weights add to 1.

Observe that $0<a_i\leq 1$.

For future reference, let us define $\alpha$ as:
\begin{eqnarray}
\alpha = \min_{i\in \scriptv}~a_i
\label{e_alpha}
\end{eqnarray}
Note that $0<\alpha\leq 1$.
Specifically, $\alpha$ is a positive constant
that is dependent only on $f$ and the graph $G(\scriptv,\scripte)$.

\end{enumerate}

~
\hrule

~

~

Similar algorithms have been proven to work correctly in
{\em fully connected} graphs \cite{AA_Dolev_1986,IBA_sync}
and {\em arbitrary directed} graphs
satisfying the necessary condition stated in
\cite{IBA_sync}.
In this paper, we provide an alternate proof of correctness
in such arbitrary graphs, using an alternate form of the
necessary condition \cite{us}.

\section{Matrix Preliminaries}

We use boldface upper case letters to denote matrices,
rows of matrices, and their elements. For instance,
$\bfH$ denotes a matrix, $\bfH_i$ denotes the $i$-th row of
matrix $\bfH$, and $\bfH_{ij}$ denotes the element at the
intersection of the $i$-th row and the $j$-th column
of matrix $\bfH$.

\begin{definition}
\label{d_stochastic}
A vector is said to be {\em stochastic} if all the elements
of the vector are {\em non-negative}, and the elements add up to 1.
A matrix is said to be row stochastic if each row of the matrix is a
stochastic vector.
\end{definition}

For a row stochastic matrix $\bfA$,
 coefficients of ergodicity $\delta(\bfA)$ and $\lambda(\bfA)$ are defined as
\cite{Wolfowitz}:
\begin{align}
\delta(\bfA) & :=   \max_j ~ \max_{i_1,i_2}~ | \bfA_{i_1\,j}-\bfA_{i_2\,j} |, \label{e_delta} \\
\lambda(\bfA) & :=  1 - \min_{i_1,i_2} \sum_j \min(\bfA_{i_1\,j} ~, \bfA_{i_2\,j}). \label{e_lambda}
\end{align}
It  is easy to see that  $0\leq \delta(\bfA) \leq 1$ and $0\leq \lambda(\bfA) \leq 1$, and that the rows are all identical if and only if $\delta(\bfA)=0$. Additionally, $\lambda(\bfA) = 0$ if and only if $\delta(\bfA) = 0$.

The next result from \cite{Hajnal58} establishes a relation between the coefficient of ergodicity $\delta(\cdot)$ of a product of row stochastic matrices, and the coefficients of ergodicity $\lambda(\cdot)$ of the individual matrices defining the product. 

\begin{claim}
\label{claim_delta}
For any $p$ square row stochastic matrices $\bfQ(1),\bfQ(2),\dots \bfQ(p)$, 
\begin{align}
\delta(\bfQ(1)\bfQ(2)\cdots \bfQ(p)) ~\leq ~
 \Pi_{i=1}^p ~ \lambda(\bfQ(i)). 
\end{align}
\end{claim}
Claim \ref{claim_delta} is proved in \cite{Hajnal58}. It implies that
if, for all $i$, $\lambda(\bfQ(i))\leq 1-\gamma$ for some $\gamma>0$, then $\delta(\bfQ(1),\bfQ(2)\cdots \bfQ(p))$ will approach zero as $p$ approaches $\infty$.

\begin{definition}
A row stochastic
 matrix $\bfH$ is said to be a {\em scrambling}\, matrix, if $\lambda(\bfH)<1$
{\normalfont \cite{Hajnal58,Wolfowitz}}.
\end{definition}

In a scrambling matrix $\bfH$, since $\lambda(\bfH)<1$, for each pair of
rows $i_1$ and $i_2$, there exists a column $j$ (which may depend on
$i_1$ and $i_2$) such that
 $\bfH_{i_1\,j}>0$ and $\bfH_{i_2\,j}>0$, and vice-versa \cite{Hajnal58,Wolfowitz}.
As a special case, if any one column of a row stochastic matrix $\bfH$
contains only non-zero elements that are lower bounded by some
constant $\gamma>0$, then $\bfH$ must be scrambling, and $\lambda(\bfH)\leq 1-\gamma$. 

\comment{======================
\begin{definition}
\label{d_type}
Two matrices of identical size are said to be of the same ``type'' if
they contain non-zero elements in identical positions.
\end{definition}
Let us denote by $\T{\bfH}$ the {\em type} of matrix $\bfH$.
A partial order can be
defined on the matrix types. Specifically, for matrices $\bfH$ and $\bfK$,
$\T{\bfH}\leq \T{\bfK}$ provided that matrix $\bfK$ is non-zero in each position
where $\bfH$ is non-zero.

\begin{lemma}
\label{l_scambling_1}
For any two row stochastic matrices $\bfH,~\bfK$ of the same size,
if $\T{\bfH}\leq \T{\bfK}$ and $\bfH$ is a scrambling matrix,
then $\bfK$ is a scrambling matrix.
\end{lemma}
\begin{proof}
Follows immediately from the definition of matrix {\em type}
and {\em scrambling} matrices.
\end{proof}

~

\begin{lemma}
\label{l_scrambling_2}
Consider a sequence $\bfH(1),\bfH(2),\cdots,\bfH(t)$ of square row stochastic matrices
with non-zero diagonals. For any subset $N$ of $\{1,2,\cdots,t\}$,
\[
\T{\Pi_{i\in N} \bfH(i)} ~ \leq ~ \T{\Pi_{1\leq i\leq t} \bfH(i)}
\]
\end{lemma}
\begin{proof}
The proof follows from the definition of matrix {\em type}, and the fact
the row stochastic matrices above have non-zero diagonals. 
\end{proof}

~
=================================================}

\section{Matrix Representation of Algorithm 1}
\label{s_claim}

Recall that $\scriptf$ is the set of faulty nodes.
Let $|\scriptf|=\phi$.
Without loss of generality, suppose that nodes 1 through $(n-\phi)$ are
fault-free, and if $\phi>0$, nodes $(n-\phi+1)$ through $n$ are faulty.

Denote by $\bfv[0]$ the column vector consisting of the initial states of
all the {\em fault-free} nodes.
Denote by $\bfv[t]$, where $t\geq 1$, the column vector consistsing of
the states of all the {\em fault-free} nodes
at the end of the $t$-th iteration, $t\geq 1$.
The $i$-th element
of vector $\bfv[t]$ is state $v_i[t]$. The size of the column
vector $\bfv[t]$ is
$(n-\phi)$.


~

\begin{claim}
\label{claim_1}
{
We can express the iterative update of the state
of a fault-free node $i$ $(1\leq i\leq n-\phi)$
performed in (\ref{e_Z}) using the matrix form in (\ref{e_matrix_i})
below,
where $\bfM_i[t]$ satisfies the following four conditions.
\begin{eqnarray}
v_i[t] & = & \bfM_i[t] ~ {\bfv}[t-1]
\label{e_matrix_i}
\end{eqnarray}
}
In addition to $t$, the row vector $\bfM_i[t]$ 
may depend on the state vector $\bfv[t-1]$ as well as the
behavior of the faulty
nodes in $\scriptf$. For simplicity, the notation $\bfM_i[t]$ does not
explicitly represent this dependence. 
\begin{enumerate}
\item $\bfM_i[t]$ is a {\em stochastic} row vector of size $(n-\phi)$.
Thus,
$\bfM_{ij}[t]\geq 0$, for $1\leq j\leq n-\phi$, and
\[
\sum_{1\leq j\leq n-\phi}~\bfM_{ij}[t] ~ = ~ 1
\]

\item $\bfM_{ii}[t]$ equals $a_i$ defined in Algorithm 1. 
Recall that $a_i\geq \alpha$.

\item $\bfM_{ij}[t]$ is non-zero
{\bf only if}  $(j,i)\in\scripte$ or $j=i$.
\item At least $|N_i^-\cap\,(\scriptv-\scriptf)| - f+1$ elements in $\bfM_i[t]$ 
are lower bounded by some constant $\beta>0$, to be defined later
($\beta$ is independent of $i$). 
Note that $N_i^-\cap\,(\scriptv-\scriptf)$ is the set of fault-free
incoming neighbors of node $i$.
\end{enumerate}

\end{claim}
\begin{proof}
The proof of this claim is presented in Section \ref{ss_claim_1} below.
The last condition above plays an important role in the proof, and the
main contribution of this paper is to ``design'' $\bfM_i[t]$ to make this
condition true.
\end{proof}

~


By ``stacking'' (\ref{e_matrix_i}) for different
$i$, $1\leq i\leq n-\phi$, we can
represent the state update for all the fault-free nodes together
using (\ref{e_matrix})
below, where $\bfM[t]$ is a $(n-\phi)\times (n-\phi)$ matrix, with its $i$-th row
being equal to $\bfM_i[t]$ in (\ref{e_matrix_i}).
\begin{eqnarray}
{\bfv}[t] & = & \bfM[t] ~ {\bfv}[t-1]
\label{e_matrix}
\end{eqnarray}
The four properties of $\bfM_i[t]$ imply that $\bfM[t]$ is a
row stochastic matrix with a non-zero diagonal.
Also, the $i$-th row of $\bfM[t]$ contains $|N_i^-\cap\,(\scriptv-\scriptf)| - f+1$
elements lower bounded by $\beta$ ($\beta$ will be defined later).
This property of $\bfM[t]$ turns out to be important in proving
convergence of Algorithm 1.

$\bfM[t]$ is said to be a {\em transition matrix}.

By repeated application of (\ref{e_matrix}), we obtain:
\begin{eqnarray*}
\bfv[t] & = & \left(\,\Pi_{i=1}^t \bfM[i]\,\right)\, \bfv[0]
\end{eqnarray*}





\subsection{Correctness of Claim \ref{claim_1}}
\label{ss_claim_1}

Figure \ref{f_sets} illustrates the various sets used here.
Some of the sets in this figure are not yet defined, and will be defined
later in the paper.
\begin{figure}
\centering
\includegraphics[width=0.7\textwidth]{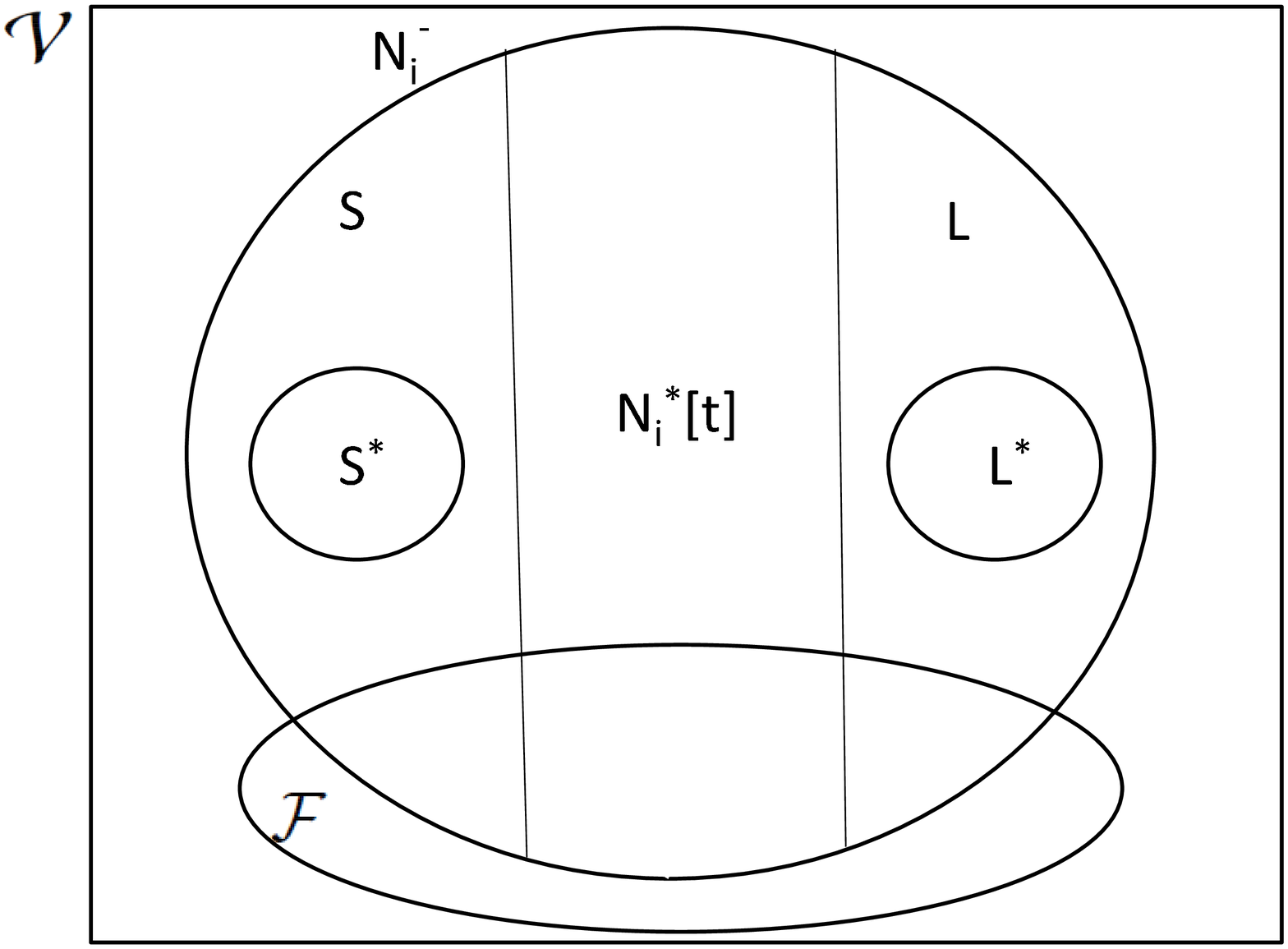}
\caption{Illustration of sets $\scriptv$, $\scriptf$, $N_i^-$,
$N_i^*[t]$, $L^*$ and $S^*$}
\label{f_sets}
\end{figure}

We prove the correctness of Claim \ref{claim_1} by constructing $\bfM_i[t]$
for $1\leq i\leq n-\phi$ that satisfies the conditions in Claim \ref{claim_1}.
Recall that nodes 1 through $n-\phi$ are fault-free, and the remaining
$\phi$ nodes ($\phi\leq f$) are faulty.

Consider a fault-free node $i$ performing the {\em Update step}
in Algorithm 1.
Recall that the largest $f$ and the smallest $f$ values are eliminated
from $r_i[t]$. Let us denote by $L$ and $S$, respectively, the set
of nodes\footnote{Although $L$ and $S$ may be different for each $t$,
for simplicity, we do not explicitly represent this dependence
on $t$ in the notations $L$ and $S$.}  
from whom the largest $f$ values and the smallest $f$ values
were received by node $i$ in iteration $t$.
Thus,
$|L|=|S|=f$, $N_i^*[t] = N_i^- - (L\cup S)$,
and $|N_i^*[t]|=|N_i^--(L\cup S)| = |N_i^-|-2f$.

For any set of nodes $X$ here, let $\delta_X$ and $g_X$ respectively
denote the number of faulty nodes, and the number of fault-free nodes, in set $X$.
For instance, $\delta_L$ and $g_L$ denote, respectively, the number
of faulty and fault-free nodes in set $L$.
Thus,
\[ \delta_L+g_L~=~\delta_S+g_S~=~f\]
Let
\[
\delta = |N_i^-\cap \scriptf|
\]
That is, the number of faulty incoming neighbors of node $i$ is denoted
as $\delta$.
Therefore, $\delta\leq\phi\leq f$,
and
\[
\delta = \delta_L+\delta_S+\deltaC
\]

Then, it follows that
\begin{eqnarray}
g_L & = & f-\delta_L ~= ~ \delta_S+\deltaC+(f-\delta), \mbox{~and}  \label{g_L}\\
g_S & = & f-\delta_S~=~\delta_L+\deltaC+(f-\delta) \label{g_S}
\end{eqnarray}

For fault-free node $i$,
we now define the elements of row $\bfM_i[t]$.
We consider two cases separately: (i) $f-\delta+\deltaC=0$,
and (ii) $f-\delta+\deltaC>0$. 

\subsubsection{$f-\delta+\deltaC=0$}
\label{ss_1}

We know that $(f-\delta) \geq 0$ and $\deltaC\geq 0$.
Therefore, $f-\delta+\deltaC=0$ implies that $f=\delta$
and $\deltaC=0$. Thus, in this case, all the nodes in $N_i^*[t]$
are fault-free.
\begin{itemize}
\item
For each 
$j \in \{i\}\cup N_i^*[t]$, define $\bfM_{ij}[t] = a_i$.
Element $\bfM_{ij}[t]$ corresponds to the term $a_iw_j$ in (\ref{e_Z}).

Recall that $a_i\geq \alpha$, and that each node in $\{i\}\cup N_i^*[t]$
in this case is fault-free.

\item 
For each $j$ such that $j\in\scriptv-\scriptf$
and $j\not\in \{i\}\cup N_i^*[t]$, define $\bfM_{ij}[t] = 0$.
\end{itemize}
Observe that with the above definition of elements of $\bfM_i[t]$,
\[ \bfM_i[t] \bfv[t-1] = \sum_{k\in \{i\}\cup N_i^*[t]} a_iw_k
\]
In the above procedure, we have set $|N_i^*[t]|+1$ elements of $\bfM_i[t]$
equal to $a_i$ (recall that $a_i\geq \alpha$).

Now, because $\delta=f$ and $|N_i^*[t]|=|N_i^-|-2f$, we have
$|N_i^-\cap(\scriptv-\scriptf)|-f+1=|N_i^-|-\delta-f+1
= |N_i^-|-2f+1 = |N_i^*[t]|+1$.
Also, in this case $a_i=1/(|N_i^*[t]|+1)$. Thus, it should
be easy to see that the conditions in Claim \ref{claim_1} are satisfied
by defining $\beta=\alpha$.

~

\subsubsection{$f-\delta+\deltaC>0$}

Since $0\leq \deltaC\leq\delta\leq f$,
$f-\delta+\deltaC>0$ implies that $f>0$.
When $f>0$, the necessary condition in \cite{IBA_sync} implies that
$|N_i^-|\geq 2f+1$. Therefore, the set $N_i^*[t]$ is non-empty.
As per (\ref{e_Z}), each node $k\in N_i^*[t]$ contributes $a_i\, w_k$
to the new state $v_i[t]$ of node $i$. We will define elements
of $\bfM_i[t]$ to account for the contribution of each node $k\in N_i^*[t]$.

Define subsets $L^*$ and $S^*$ such that
$L^*\subseteq L$, $S^*\subseteq S$, $L^*\cap\scriptf=S^*\cap\scriptf=\Phi$,
and $|L^*|=|S^*|=f-\delta+\deltaC$. That is, sets $L^*$ and $S^*$
are subsets of $L$ and $S$, respectively, each of size $f-\delta+\deltaC$,
and containing only fault-free nodes.
Expressions (\ref{g_L}) and (\ref{g_S}) for $g_L$ and $g_S$
imply that such subsets exist.

Let 
\[ L^*=\{l_j~ | ~ 1\leq j\leq f-\delta+\deltaC\}\]
and \[ S^*=\{s_j~ | ~ 1\leq j\leq f-\delta+\deltaC\}.\]
Consider any node $k\in N_i^*[t]$.
For each $j$, $1\leq j\leq f-\delta+\deltaC$,
\[ v_{s_j}[t-1] \leq w_k \leq v_{l_j}[t-1] \]
Therefore,
we can find weights $\lambda_{k,j}\geq 0$ and $\psi_{k,j}\geq 0$
 such that
\[
\lambda_{k,j} +  \psi_{k,j}~=~1\] and
\[
w_k~=~ ~\lambda_{k,j}\, v_{l_j}[t-1] ~ + ~
        \psi_{k,j} \, v_{s_j}[t-1] 
\]
Clearly, at least one of the weights $\lambda_{k,j}$ and $\psi_{k,j}$
must be $\geq 1/2$.
Now, observe that
\begin{eqnarray}
a_i\, w_k & = &  
\frac{a_i}{f-\delta+\deltaC}  ~
\sum_{1\leq j\leq f-\delta+\deltaC} 
\left(
\lambda_{k,j}\, v_{l_j}[t-1] +
        \psi_{k,j} \, v_{s_j}[t-1]
\right)
\label{e_faulty} 
\end{eqnarray}
The above equality is true independent of whether $k$ is fault-free or faulty.
We will later use the above equality for the case when $k$ is a faulty node.
When $k$ is fault-free, \[w_k=v_k[t-1],\] and we can similarly obtain the equality below.
\begin{eqnarray}
a_iw_k & = & \frac{a_i}{2}v_k[t-1]
             ~+~ 
	  \frac{a_i}{2(f-\delta+\deltaC)} ~
\sum_{1\leq j\leq f-\delta+\deltaC} 
\left(
\lambda_{k,j}\, v_{l_j}[t-1]  + 
        \psi_{k,j} \, v_{s_j}[t-1]
\right) \nonumber\\
 \label{e_faultfree}
\end{eqnarray}

We now use (\ref{e_Z}), (\ref{e_faulty}) and (\ref{e_faultfree}) to define elements
of $\bfM_i[t]$ in the following four cases:
\begin{itemize}
\item {\bf Case 1: Node $i$}\\ Define $\bfM_{ii}[t]=a_i$. This is obtained by observing in (\ref{e_Z}) that
the contribution of node $i$ to the new state $v_i[t]$ is $a_iw_i=a_iv_i[t-1]$.
\item {\bf Case 2: Fault-free nodes in $N_i^*[t]$}\\
For each $k\in N_i^*[t]\cap (\scriptv-\scriptf)$,
define $\bfM_{ik}[t] = \frac{a_i}{2}$. This choice is motivated by (\ref{e_faultfree})
wherein the contribution of node $k$ to $a_iw_k$ is $\frac{a_i}{2}v_k[t-1]$.
In Case 2, $|N_i^*[t]\cap(\scriptv-\scriptf)|=|N_i^-|-\delta$ elements
of $\bfM_i[t]$ are defined.
\item {\bf Case 3: Nodes in $L^*$ and $S^*$}\\ For $1\leq j\leq f-\delta+\deltaC$, consider $l_j\in L^*$. In this
case, 
\begin{eqnarray*}
\bfM_{il_j}[t] &=&
\sum_{k\in N_i^*[t] \cap \scriptf} \frac{a_i}{f-\delta+\deltaC} \lambda_{k,j}
~+~
\sum_{k\in N_i^*[t]\cap (\scriptv-\scriptf)}
\frac{a_i}{2(f-\delta+\deltaC)} \lambda_{k,j}
\end{eqnarray*}
Similarly, for $1\leq j\leq f-\delta+\deltaC$, consider $s_j\in S^*$. In this
case, 
\begin{eqnarray*}
\bfM_{is_j}[t] &=&
\sum_{k\in N_i^*[t] \cap \scriptf} \frac{a_i}{f-\delta+\deltaC} \psi_{k,j}
~+~
\sum_{k\in N_i^*[t]\cap (\scriptv-\scriptf)}
\frac{a_i}{2(f-\delta+\deltaC)} \psi_{k,j}
\end{eqnarray*}
These expressions are obtained by summing (\ref{e_faulty})
and (\ref{e_faultfree}), respectively, over the faulty and fault-free nodes
in $N_i^*[t]$, and then identifying the contribution
of each node in $L^*$ and $S^*$ to this sum.
Recall the earlier observation that at least one of $\lambda_{k,j}$
and $\psi_{k,j}$ must be $\geq 1/2$ for each pair $k,j$ where $k\in N_i^*[t]$
and $1 \leq j\leq f-\delta+\deltaC$.
Therefore, it follows that
at least $f-\delta+\deltaC$ elements of $\bfM_i[t]$ defined
in Case 3 must be $\geq \frac{a_i}{4(f-\delta+\deltaC)}$.

\item {\bf Case 4: Nodes in $(\scriptv-\scriptf)-(\{i\}\cup N_i^*[t]\cup L^*\cup S^*)$}\\
These fault-free nodes have not yet been considered in Cases 1, 2 and 3.
For each node $k\in (\scriptv-\scriptf)-(\{i\}\cup N_i^*[t]\cup L^*\cup S^*)$,
we assign $\bfM_{ik}[t]=0$. 
\end{itemize}
Observe that above the definition of the elements of $\bfM_i[t]$
ensures that 
\[ \sum_{j\in\{i\}\cup N_i^*[t]}a_iw_j~=~\bfM_i[t]\bfv[t-1]\]
However, the contribution by the
faulty nodes in $N_i^*[t]$ in (\ref{e_Z}) is
now replaced by an
equivalent contribution by the nodes in $L^*$ and $S^*$.

Now let us verify that the four conditions in Claim \ref{claim_1} hold
for the above assignments to the elements of $\bfM_i[t]$.
\begin{enumerate}
\item
Observe that all the elements of $\bfM_i[t]$ are non-negative.
Case 1 specifies just $\bfM_{ii}[t]=a_i$.
The elements of $\bfM_i[t]$ specified in Case 2 
add up to
\[
\frac{a_i}{2} ~ |N_i^*[t]\cap (\scriptv-\scriptf)|
\]
Recall that for each $j$, $1\leq j\leq (f-\delta+\deltaC)$,
$\lambda_{k,j}+\psi_{k,j}=1$ for $k\in N_i^*[t]$.
Therefore, when added over all $k\in N_i^*[t]$
and $1\leq j\leq (f-\delta+\deltaC)$, the elements of $\bfM_i[t]$ specified in Case 3 
add up to
\[
a_i ~ |N_i^*[t]\cap \scriptf| ~+~
\frac{a_i}{2} ~ |N_i^*[t]\cap (\scriptv-\scriptf)|
\]
Therefore, when all the elements of $\bfM_i[t]$ defined in Cases 1,
2 and 3 are added together, we get
\begin{eqnarray*}
a_i ~+~ a_i\, |N_i^*[t]\cap\scriptf|
~+~ 
a_i  |N_i^*[t]\cap (\scriptv-\scriptf)| ~=~
 a_i (|N_i^*[t]|+1) ~=~1
&&
\end{eqnarray*}
because 
$a_i=1/(|N_i^*[t]|+1)$.
Now observe that the elements specified in Cases 1, 2 and 3 are clearly
$\leq 1$.
In the expression for $\bfM_{il_j}[t]$ in Case 3, observe that the two
summations on the right side together contain $|N_i^*[t]|$ terms,
and in these terms, observe that $\lambda_{k,j}\leq 1$,
$f-\delta+\deltaC\geq 1$ and $a_i=\frac{1}{|N_i^*[t]|+1}$. Therefore,
$\bfM_{il_j}[t]<1$.
Similarly, we can show that $\bfM_{is_j}[t]<1$ as well.

Thus, we have shown that $\bfM_i[t]$ is a stochastic vector.

\item $\bfM_{ii}[t]=a_i$ as specified in Case 1.

\item Since $\bfM_{ij}[t]$ is defined to be non-zero only in Cases 1, 2 and 3,
which consider the nodes only in $\{i\}\cup N_i^-$, it follows
that $\bfM_{ij}[t]$ is non-zero {\em only if} $(j,i)\in\scripte$
or $j=i$.

\item
Cases 1 and 2 together set $1+|N_i^*[t]\cap(\scriptv-\scriptf)|=
1+|N_i^*[t]|-\deltaC$
elements of $\bfM_i[t]$ to be $\geq a_i/2$. We observed
earlier that Case 3 results in at least
$f-\delta+\deltaC$ elements of $\bfM_i[t]$ 
being $\geq \frac{a_i}{4(f-\delta+\deltaC)}$.
Also, observe that the elements of $\bfM_i[t]$ specified
in Cases 1 and 2 are distinct from those specified in Case 3,
and that $\frac{a_i}{2} \geq \frac{a_i}{4(f-\delta+\deltaC)}$. Thus, overall, at
least
\begin{eqnarray*}
(1+|N_i^*[t]|-\deltaC) ~+~
f-\delta+\deltaC ~=~ |N_i^*[t]|+f-\delta+1
~=~ |N_i^-|-f-\delta+1 && \\ ~=~ |N_i^-\cap(\scriptv-\scriptf)|-f-1
&& \end{eqnarray*}
elements of $\bfM_i[t]$ are set
$\geq \frac{a_i}{4(f-\delta+\deltaC)}$.
Derivation of the above equation uses the facts that
$|N_i^*[t]|=|N_i^-|-2f$ and
$|N_i^-\cap(\scriptv-\scriptf)|=|N_i^-|-\delta$. Then
by defining $\beta$ as below, condition 4 in Claim \ref{claim_1}
holds true.
\[
\beta = \frac{\alpha}{4(f-\delta+\deltaC)}
\]
\end{enumerate}

Therefore, Claim \ref{claim_1} is proved correct.

~

\subsection{Correspondence Between Sufficiency Condition and $\bf M[t]$}

Let us define set $R_\scriptf$ of subgraphs of $G(\scriptv,\scripte)$ as follows.
\begin{eqnarray}
R_\scriptf & = & \{ H ~ | ~ H \mbox{~ is obtained by removing all the faulty}
\nonumber \\
&& \mbox{ nodes
	 from $\scriptv$ along with their edges, and then} \nonumber
	\\ && \mbox{ removing any additional $f$ incoming edges}\nonumber \\ 
&&
	\mbox{ at each fault-free node} \}
\end{eqnarray}
Thus, $\scriptv-\scriptf$ is the set of nodes in each graph in $R_\scriptf$.

Let $\tau$ denote $|R_\scriptf|$. 
$\tau$ depends on $\scriptf$ and the underlying network, and it is finite.

\begin{claim}
\label{claim_suff}
Suppose that graph $G(\scriptv,\scripte)$ satisfies the necessary condition in 
Theorem 2 in \cite{us}. Then it follows that in each $H\in R_\scriptf$,
there exists at least one node that has directed paths to all the nodes in $H$
(consisting of the edges in $H$).
\end{claim}
\begin{proof}
The proof follows from Theorem 2 of \cite{us}.
\end{proof}

In this discussion, let us denote a graph by an italic upper case letter,
and the corresponding {\em connectivity matrix} using the same letter
in boldface upper case. Thus,
$\bfH$ will denote the connectivity matrix for graph $H\in R_\scriptf$;
$\bfH$ is defined as follows:
(i) for $1\leq i,j\leq n-\phi$,
if there is a directed link from node $j$ to node $i$ in
graph $H$ then $\bfH_{ij}=1$,
and (ii) $\bfH_{ii}=1$ for $1\leq i\leq n-\phi$.
Note that in our notation, the $i$-th row of $\bfH$ (that is, $\bfH_i$)
corresponds to the incoming links at node $i$, and the self-loop
at node $i$.
The connectivity matrix $\bfH$ for any $H\in R_\scriptf$
has a non-zero diagonal.


\begin{lemma}
\label{l_one_column}
For any $H\in R_\scriptf$, $\bfH^{n-\phi}$ has at least one non-zero column.
\end{lemma}
\begin{proof}
By Claim \ref{claim_suff}, in graph $H$ there exists at least one node, say
node $k$, that has a directed path in $H$ to all the remaining nodes in $H$.
Since the length of the path from $k$ to any other node in $H$ can contain
at most $n-\phi-1$ directed edges,
the $k$-th column of matrix ${\bfH^{n-\phi}}$ will
be non-zero.\footnote{That is, all the elements of the column will be
non-zero (more precisely, positive, since the elements of matrix $\bfH$
are non-negative).
Also, such a non-zero column will exist in $\bfH^{n-\phi-1}$ too.
We use the loose bound of $n-\phi$ to simplify the presentation. }
\end{proof}

~

\begin{definition}
We will say that an element of a matrix is ``non-trivial'' if it is lower
bounded by $\beta$.
\end{definition}

\begin{definition}
For matrices $\bfA$ and $\bfB$ of identical size, and
a scalar $\gamma$, $\bfA\leq \gamma \, \bfB$ provided
that $\bfA_{ij}\leq \gamma\, \bfB_{ij}$ for all $i,j$.
\end{definition}

\begin{lemma}
\label{l_H}
For any $t\geq 1$, there exists a graph $H[t]\in R_\scriptf$ such that
$\beta \, \bfH[t] ~ \leq ~  {\bfM[t]}$.
\end{lemma}
\begin{proof}
Observe that the $i$-th row of the transition matrix $\bfM[t]$ corresponds to the state update
performed at fault-free node $i$. Recall from Claim \ref{claim_1} that
the $\bfM_{ij}$ is non-zero {\bf only if} link $(j,i)\in\scripte$.
Also, by Claim \ref{claim_1},
$\bfM_i[t]$ (i.e., the $i$-th row of $\bfM[t]$)
contains at least $|N_i^-\cap\,(\scriptv-\scriptf)| - f+1$
{\em non-trivial} elements corresponding
to {\bf fault-free} incoming neighbors of node $i$ and itself (i.e., the
diagonal element).

Now observe that, for any subgraph $H\in R_\scriptf$, $i$-th row of 
$\bfH$ contains exactly $|N_i^-\cap\,(\scriptv-\scriptf)| - f+1$ non-zero elements,
including the diagonal element.

Considering the above two observations, and the definition of set $R_\scriptf$,
the lemma follows.
\end{proof}

~

\section{Correctness of Algorithm 1}

The proof below uses techniques also applied in prior work
(e.g., \cite{jadbabaie_consensus,Benezit,vaidyaII,Zhang}),
with some similarities to the arguments used in \cite{vaidyaII,Zhang}.

\begin{lemma}
\label{l_product_H}
In the product below of $\bfH[t]$ matrices for consecutive
$\tau(n-\phi)$ iterations, at least one column is non-zero. 
\[
\Pi_{t=z}^{z+\tau(n-\phi)-1} \, \bfH[t]
\]
\end{lemma}
\begin{proof}
Since the above product consists of $\tau(n-\phi)$ matrices
in $R_\scriptf$,
at least one of the $\tau$ distinct connectivity matrices
in $R_\scriptf$, say matrix $\bfH_*$, will appear in the above
product at least $n-\phi$ times.

Now observe that: (i)
By Lemma \ref{l_one_column}, $\bfH_*^{n-\phi}$ contains a non-zero
column, say the $k$-th column is non-zero,
and (ii) all the $\bfH[t]$ matrices in the product contain a non-zero diagonal.
These two observations together imply that the $k$-th column in the above product 
is non-zero.
\end{proof}

Let us now define a sequence of matrices $\bfQ(i)$ such that
each of these matrices is a product of $\tau(n-\phi)$ of the
$\bfM[t]$ matrices. Specifically,
\[
\bfQ(i) ~=~ \Pi_{t=(i-1)\tau(n-\phi)+1}^{i\tau(n-\phi)} ~ \bfM[t]
\]
Observe that
\begin{eqnarray}
\bfv[k\tau(n-\phi)] & = & \left(\, \Pi_{i=1}^k ~ \bfQ(i) \,\right)~\bfv[0]
\end{eqnarray}

\begin{lemma}
\label{l_Q}
For $i\geq 1$, $\bfQ(i)$ is a scrambling row stochastic matrix,
and $\lambda(\bfQ(i))$ is bounded from above by a constant
smaller than 1.
\end{lemma}
\begin{proof}

$\bfQ(i)$ is a product of row stochastic matrices ($\bfM[t]$), therefore,
$\bfQ(i)$ is row stochastic.

From Lemma \ref{l_H}, for each $t$,
\[
\beta \, \bfH[t] ~ \leq ~ \bfM[t]
\]
Therefore, 
\[
\beta^{\tau(n-\phi)} ~ \Pi_{t=(i-1)\tau(n-\phi)+1}^{i\tau(n-\phi)} ~ \bfH[t] ~ \leq 
~ \bfQ(i)
\]
By using $z=(i-1)(n-\phi)+1$ in Lemma \ref{l_product_H},
we conclude that the matrix product on the left side
of the above inequality contains a non-zero column. Therefore, $\bfQ(i)$ contains
a non-zero column as well. Therefore, $\bfQ(i)$ is a scrambling matrix.

Observe that $\tau(n-\phi)$ is finite, therefore, $\beta^{\tau(n-\phi)}$
is non-zero. Since the non-zero terms in $\bfH[t]$ matrices are all 1,
the non-zero elements in $\Pi_{t=(i-1)\tau(n-\phi)+1}^{i\tau(n-\phi)} \bfH[t]$
must each be $\geq$ 1. Therefore, there exists a non-zero column in $\bfQ(i)$
with all the elements in the column being $\geq \beta^{\tau(n-\phi)}$.
Therefore $\lambda(\bfQ(i))\leq 1-\beta^{\tau(n-\phi)}$.
\end{proof}

\begin{theorem}
\label{t}
Algorithm 1 satisfies the validity and the convergence conditions.
\end{theorem}
\begin{proof}

Since $\bfv[t]=\bfM[t]\,v[t-1]$, and $\bfM[t]$ is a row stochastic matrix, it
follows that
Algorithm 1 satisfies the validity condition.

By Claim \ref{claim_delta}, 
\begin{eqnarray}
\lim_{t\rightarrow \infty} \delta(\Pi_{i=1}^t \bfM[t])
& \leq & \lim_{t\rightarrow\infty} \Pi_{i=1}^t \lambda(\bfM[t]) \\ 
& \leq & \lim_{i\rightarrow\infty} \Pi_{i=1}^{\lfloor\frac{t}{\tau(n-\phi)}\rfloor} \lambda(\bfQ(i)) \\
& = & 0 
\end{eqnarray}
The above argument makes use of the facts that
$\lambda(\bfM[t])\leq 1$ and $\lambda(\bfQ(i))\leq (1-\beta^{\tau(n-\phi)})<1$.
Thus, the rows of $\Pi_{i=1}^t \bfM[t]$ become identical in the limit.
This observation, and the fact that $\bfv[t]=(\Pi_{i=1}^t \bfM[i])\bfv[t-1]$ together imply that
the state of the fault-free nodes satisfies the
convergence condition.

Now, the validity and convergence conditions
together imply that 
there exists a positive scalar $c$ such that
\[
\lim_{t\rightarrow\infty}
\bfv[t] ~ = ~ \lim_{t\rightarrow\infty} \left( \Pi_{i=1}^t \bfM[i]) \right)\,
\bfv[0] ~ = ~ c\,{\bf 1}
\] 
where {\bf 1} denotes a column with all its elements being 1.

\end{proof}

\section{Extension of Above Results}
\label{s_extend}

In this paper, we
analyzed IABC Algorithm 1 designed for synchronous systems. Similar
analysis also applies for IABC Algorithm 2 presented
in \cite{us} for asynchronous systems. 

The analysis will also naturally extend to an IABC algorithm for the 
{\em partially synchronous algorithmic} model
presented in \cite{AA_convergence_markov}, which assumes a bounded
delay in propagation of state between neighbors, and a bounded delay between
consecutive state updates at each node in the network.
The generalization of Algorithm 1 to the
{\em partially synchronous algorithmic} model will allow a node $i$,
if performing state update in iteration $t$, to form vector $r_i[t]$
using the most recent known states of its incoming neighbors; these states
of the neighbors may correspond to any of the prior $B$ iterations,
for some bounded $B$.
A similar IABC algorithm can also be used 
in time-varying network topologies (i.e., networks wherein the set
of links available in iteration $t$ varies with $t$);
the above analysis will
then extend to time-varying topologies as well, with the
algorithm performing correctly so long as
the connectivity matrices for the graphs at different $t$ 
jointly satisfy some reasonable properties, as in \cite{jadbabaie_consensus,Benezit,vaidyaII}.


\section{Summary}

We presented a proof of validity and convergence of Algorithm 1 by
expressing the algorithm in the matrix form. The main contribution of
the paper is to express the algorithm in matrix form that allows us to
prove its convergence under certain necessary conditions on the underlying
communication graph. Thus, the proof implies that the necessary conditions
are also sufficient. 
The key to the proof is to ``design'' the transition matrix for each
iteration such that each row of the matrix contains a large enough number
of elements that are bounded away from 0.


\begin{thebibliography}{10}

\bibitem{AA_PFCN}
A.~Azadmanesh and H.~Bajwa.
\newblock Global convergence in partially fully connected networks (pfcn) with
  limited relays.
\newblock In {\em Industrial Electronics Society, 2001. IECON '01. The 27th
  Annual Conference of the IEEE}, volume~3, pages 2022 --2025 vol.3, 2001.

\bibitem{AA_async_PCN}
M.~H. Azadmanesh and R.~Kieckhafer.
\newblock Asynchronous approximate agreement in partially connected networks.
\newblock {\em International Journal of Parallel and Distributed Systems and
  Networks}, 5(1):26--34, 2002.
\newblock http://ahvaz.unomaha.edu/azad/pubs/ijpdsn.asyncpart.pdf

\bibitem{Benezit}
F. Benezit, V. Blondel, P. Thiran, J. Tsitsiklis, and M. Vetterli, “Weighted gossip: Distributed averaging using non-doubly stochastic
matrices,” in Proc. of IEEE International Symposium on Information Theory, June 2010, pp. 1753--1757.

\bibitem{AA_convergence_markov}
D.~P. Bertsekas and J.~N. Tsitsiklis.
\newblock {\em Parallel and Distributed Computation: Numerical Methods}.
\newblock Optimization and Neural Computation Series. Athena Scientific, 1997.

\bibitem{AA_Dolev_1986}
D.~Dolev, N.~A. Lynch, S.~S. Pinter, E.~W. Stark, and W.~E. Weihl.
\newblock Reaching approximate agreement in the presence of faults.
\newblock {\em J. ACM}, 33:499--516, May 1986.


\bibitem{AA_Fekete_aoptimal}
A.~D. Fekete.
\newblock Asymptotically optimal algorithms for approximate agreement.
\newblock In {\em Proceedings of the fifth annual ACM symposium on Principles
  of distributed computing}, PODC '86, pages 73--87, New York, NY, USA, 1986.
  ACM.

\bibitem{FLP_one_crash}
M.~J. Fischer, N.~A. Lynch, and M.~S. Paterson.
\newblock Impossibility of distributed consensus with one faulty process.
\newblock {\em J. ACM}, 32:374--382, April 1985.

\bibitem{Hajnal58}
J. Hajnal, ``Weak ergodicity in non-homogeneous Markov chains,” Proceedings of the Cambridge Philosophical Society, vol. 54, pp.
pp. 233--246, 1958.

\bibitem{jadbabaie_consensus}
A.~Jadbabaie, J.~Lin, and A.~Morse.
\newblock Coordination of groups of mobile autonomous agents using nearest
  neighbor rules.
\newblock {\em Automatic Control, IEEE Transactions on}, 48(6):988--1001,
  june 2003.

\bibitem{AA_PCN_Local}
R.~M. Kieckhafer and M.~H. Azadmanesh.
\newblock Low cost approximate agreement in partially connected networks.
\newblock {\em Journal of Computing and Information}, 3(1):53--85, 1993.


\bibitem{diss_Sundaram}
H.~LeBlanc, H.~Zhang, S.~Sundaram, and X.~Koutsoukos.
\newblock Consensus of multi-agent networks in the presence of adversaries
  using only local information.
\newblock {\em HiCoNs}, 2012.

\bibitem{Leblanc_HSCC_1}
H.~LeBlanc and X.~Koutsoukos.
\newblock Consensus in networked multi-agent systems with adversaries.
\newblock {\em 14th International conference on Hybrid Systems: Computation and
  Control (HSCC)}, 2011.

\bibitem{Leblanc_HSCC_2}
H.~LeBlanc and X.~Koutsoukos.
\newblock Low complexity resilient consensus in networked multi-agent systems
  with adversaries.
\newblock {\em 15th International conference on Hybrid Systems: Computation and
  Control (HSCC)}, 2012.

\bibitem{AA_nancy}
N.~A. Lynch.
\newblock {\em Distributed Algorithms}.
\newblock Morgan Kaufmann, 1996.


\bibitem{IBA_sync}
N.~H. Vaidya, L.~Tseng, and G.~Liang.
\newblock Iterative approximate Byzantine consensus in arbitrary directed
  graphs.
\newblock {\em CoRR}, abs/1201.4183v1 (January 2012), abs/1201.4183v2 (Februar 2012).
Available from {\tt http://arxiv.org}.

\bibitem{us}
N.~H. Vaidya, L.~Tseng, and G.~Liang.
\newblock Iterative approximate Byzantine consensus in arbitrary directed
  graphs -- Part II: Synchronous and asynchronous systems.
\newblock Technical report, University of Illinois at Urbana-Champaign,
  February 2012.
\newblock {\tt http://www.crhc.illinois.edu/wireless/papers/\\approx\_consensus\_II.pdf}

\bibitem{vaidyaII}
N. H. Vaidya, C. N. Hadjicostis, A. D. Dominguez-Garcia.
Distributed Algorithms for Consensus and Coordination in the Presence of Packet-Dropping Communication Links - Part II: Coefficients of Ergodicity Analysis Approach. September 2011. Available from {\tt http://arxiv.org/abs/1109.6392}.

\bibitem{Wolfowitz}
J. Wolfowitz, “Products of indecomposable, aperiodic, stochastic matrices,” Proceedings of the American Mathematical Society,
vol. 14, no. 5, pp. pp. 733--737, 1963.

\bibitem{Zhang}
H. Zhang and S. Sundaram. Robustness of Information Diffusion Algorithms to Locally Bounded Adversaries. In CoRR, volume abs/1110.3843, 2011. {\tt http://arxiv.org/abs/1110.3843}

\end{thebibliography}
\end{document}